\documentclass[english]{article}
\usepackage[T1]{fontenc}
\usepackage[latin9]{inputenc}
\usepackage{verbatim}
\usepackage{url}
\usepackage{amsthm}
\usepackage{amsmath}
\usepackage{amssymb}

\providecommand{\tabularnewline}{\\}

\theoremstyle{plain}
\theoremstyle{plain}
\newtheorem{thm}{Theorem}
  \theoremstyle{plain}
  \newtheorem{cor}[thm]{Corollary}
  \theoremstyle{remark}
  \newtheorem{claim}[thm]{Claim}
  \theoremstyle{plain}
  \newtheorem{lem}[thm]{Lemma}
  \newcounter{casectr}
  \newenvironment{caseenv}
  {\begin{list}{{\itshape\ Case} \arabic{casectr}.}{%
   \setlength{\leftmargin}{\labelwidth}
   \addtolength{\leftmargin}{\parskip}
   \setlength{\itemindent}{\listparindent}
   \setlength{\itemsep}{\medskipamount}
   \setlength{\topsep}{\itemsep}}
   \setcounter{casectr}{0}
   \usecounter{casectr}}
  {\end{list}}
  \theoremstyle{plain}
  \newtheorem{prop}[thm]{Proposition}

\usepackage{fullpage}

\usepackage{babel}

\begin{document}

\title{3/2 firefighters are not enough}

\author{Ohad N.~Feldheim%
\thanks{School of Mathematics, Raymond and Beverly Sackler Faculty of Exact
Sciences, Tel Aviv University, Tel Aviv, Israel. E-mail: ohad\_f@netvision.net.il.
Research supported by an ERC advanced grant.%
} \and Rani Hod%
\thanks{School of Computer Science, Raymond and Beverly Sackler Faculty of
Exact Sciences, Tel Aviv University, Tel Aviv, Israel. E-mail: rani.hod@cs.tau.ac.il.
Research supported by an ERC advanced grant.%
}}
\maketitle
\begin{abstract}
The firefighter problem is a monotone dynamic process in graphs that
can be viewed as modeling the use of a limited supply of vaccinations
to stop the spread of an epidemic. In more detail, a fire spreads
through a graph, from burning vertices to their unprotected neighbors.
In every round, a small amount of unburnt vertices can be protected
by firefighters. How many firefighters per turn, on average, are needed
to stop the fire from advancing?

We prove tight lower and upper bounds on the amount of firefighters
needed to control a fire in the Cartesian planar grid and in the strong
planar grid, resolving two conjectures of Ng and Raff.
\end{abstract}

\section{Introduction}

The firefighter problem is the following dynamic problem introduced
by Hartnell~\cite{Hartnell95}. Given an undirected graph $G=\left(V,E\right)$,
a fire initially breaks out at a nonempty subset of vertices $\varnothing\subset S\subset V$.
In every round $t$, $f\left(t\right)$ firefighters are available
to be positioned at vacant and unburnt vertices of $G$. These firefighters
remain on their assigned vertices for the entire process, protecting
them from the fire. At the end of each round, the fire spreads to
all unprotected vertices adjacent to at least one burnt vertex.

For infinite graphs, two scenarios are possible:
\begin{enumerate}
\item [(\emph{i})] In finite time, the fire is controlled (i.e., is unable
to spread further) and thus all but a finite number of vertices remain
unburnt and unprotected.
\item [(\emph{ii})] The fire spreads indefinitely.
\end{enumerate}
Natural questions that can be asked are whether the fire can be controlled,
and, if so, how fast; a related question is how many vertices can
we save: absolute number for finite graphs, measure (defined properly)
for infinite graphs.

The firefighter problem was considered for a variety of families of
graphs, including infinite grids~\cite{DH07,Fogarty03-msc,Messin07,Messin08,Messin*,WM02},
finite grids~\cite{MW03,WM02}, and trees~\cite{FKMR07,Hartnell95}.

\bigskip{}

\global\long\def\cartesian{\,\square\,}
\global\long\def\triangular{\,\triangle\,}
\global\long\def\strong{\boxtimes}

In this paper we focus on two infinite grids: the Cartesian grid $\mathbb{Z}\cartesian\mathbb{Z}$,
which is the $4$-regular graph on the vertex set $\mathbb{Z}\times\mathbb{Z}$
in which the neighbors of every vertex form a sphere of radius $1$
with respect to the $\ell_{1}$ metric, and the strong grid $\mathbb{Z}\strong\mathbb{Z}$,
which is the $8$-regular graph on the vertex set $\mathbb{Z}\times\mathbb{Z}$
in which the neighbors of every vertex form a sphere of radius $1$
with respect to the $\ell_{\infty}$ metric. A third infinite grid,
which we only briefly mention, is the 6-regular triangular grid $\mathbb{Z}\triangular\mathbb{Z}$
satisfying $\mathbb{Z}\cartesian\mathbb{Z}\subset\mathbb{Z}\triangular\mathbb{Z}\subset\mathbb{Z}\strong\mathbb{Z}$.

We refer henceforth to vertices of these grids as points.

\subsection{Previous results}

Wang and Moeller~\cite{WM02} proved that when $f\equiv1$, a single-source
fire cannot be controlled even in the non-negative quadrant $\mathbb{N}\cartesian\mathbb{N}$
of $\mathbb{Z}\cartesian\mathbb{Z}$. With an additional firefighter
($f\equiv2$) a single-source fire in $\mathbb{Z}\cartesian\mathbb{Z}$
can be controlled within 8 turns and 18 burnt points. Fogarty~\cite{Fogarty03-msc}
proved that with $f\equiv2$ firefighters, any finite-source fire
in $\mathbb{Z}\cartesian\mathbb{Z}$ can be controlled. Messinger~\cite{Messin08}
proved that for any $n\in\mathbb{N}$, a single-source fire in $\mathbb{Z}\cartesian\mathbb{Z}$
can be controlled using the periodic function \[
f\left(t\right)=\begin{cases}
2, & t\bmod\left(2n+1\right)\textrm{ is zero or odd};\\
1, & t\bmod\left(2n+1\right)\textrm{ is even and nonzero},\end{cases}\]
whose average is $\left(3n+2\right)/\left(2n+1\right)=3/2+O\left(1/n\right)$.
Ng and Raff~\cite{NR08} proved that any periodic function $f$ whose
average exceeds $3/2$ allows the firefighters to control any finite-source
fire in $\mathbb{Z}\cartesian\mathbb{Z}$.

Develin and Hartke~\cite{DH07} proved that, for $d\ge3$, a single-source
fire in $\mathbb{Z}^{\square d}=\mathbb{Z}\cartesian\cdots\cartesian\mathbb{Z}$
cannot be controlled using $f\equiv2d-2$ firefighters (and is controlled
by $f\equiv2d-1$ firefighters within just two turns). Moreover, they
showed that for any fixed $m$, $f\equiv m$ firefighters cannot control
an $m^{2}$-source fire in $\mathbb{Z}^{\square d}$.

Fogarty~\cite{Fogarty03-msc} claimed that $f\equiv2$ firefighters
cannot control a single-source fire in the triangular grid $\mathbb{Z}\triangular\mathbb{Z}$
but her proof is not complete. Messinger~\cite{Messin07} proved
that slightly more firefighters can control it; namely, for any $n\in\mathbb{N}$
she describes a strategy using $f\left(t\right)=\begin{cases}
3, & t=0\bmod n;\\
2, & t\neq0\bmod n\end{cases}$ firefighters.

Messinger~\cite{Messin*} claimed that $f\equiv3$ firefighters cannot
control a single-source fire in the strong grid $\mathbb{Z}\strong\mathbb{Z}$,
or even to restrain it to a single quadrant, but here, too, the proof
is not complete. She proved that slightly more firefighters can control
it; that is, for any $n\in\mathbb{N}$ her scheme needs only $f\left(t\right)=\begin{cases}
4, & t=0\bmod n;\\
3, & t\neq0\bmod n\end{cases}$ firefighters.

\subsection{Our results}

All of our results depend on properties of the cumulative sum $f^{*}\left(t\right)=\sum_{\tau=1}^{t}f\left(\tau\right)$
of the function $f$.

We show the following lower bound for the Cartesian grid $\mathbb{Z}\cartesian\mathbb{Z}$,
closing the gap between the existing lower bound $1$ and the upper
bound $3/2+\epsilon$. 
\begin{thm}
\label{thm:cartesian-grid-LB}If $f^{*}\left(t\right)$ never exceeds
$\left(3t+1\right)/2$ then no strategy using $f$ firefighters can
control a single-source fire in $\mathbb{Z}\cartesian\mathbb{Z}$.
\end{thm}
Theorem~\ref{thm:cartesian-grid-LB} settles \cite[Conjecture 1]{NR08}
when applied to the function $f\left(t\right)=1+\left(t\bmod2\right)$
--- that is, the sequence $2,1,2,1,\ldots$. Moreover, Theorem~\ref{thm:cartesian-grid-LB}
implies the lower bound $3$ for the strong grid $\mathbb{Z}\boxtimes\mathbb{Z}$.
\begin{cor}
\label{cor:strong-grid-LB}If $f^{*}\left(t\right)$ never exceeds
$3t+1$ then no strategy using $f$ firefighters can control a single-source
fire in $\mathbb{Z}\strong\mathbb{Z}$.
\end{cor}
We show a essentially matching upper bound for the strong grid.
\begin{thm}
\label{thm:strong-grid-UB}If $\lim\inf f^{*}\left(t\right)/t>3$
then for any finite-source fire in $\mathbb{Z}\strong\mathbb{Z}$,
there exists a strategy using $f$ firefighters that can control it.
\end{thm}
Theorem~\ref{thm:strong-grid-UB} yields the following generalization
of the known upper bound for the Cartesian grid $\mathbb{Z}\cartesian\mathbb{Z}$,
which allows for non-periodic functions. This settles \cite[Conjecture 2]{NR08}.
\begin{cor}
\label{cor:cartesian-grid-UB}If $\lim\inf f^{*}\left(t\right)/t>3/2$
then for any finite-source fire in $\mathbb{Z}\cartesian\mathbb{Z}$,
there exists a strategy using $f$ firefighters that can control it.
\end{cor}
Note that $\lim\inf$ is the correct measure for $f^{*}\left(t\right)/t$
rather than $\lim\sup$, since it is easy to build, for any $\epsilon>0$,
an example of a function $f$ satisfying $\lim\sup f^{*}\left(t\right)/t=4-\epsilon$
(resp., $8-\epsilon$) such that a single-source fire in $\mathbb{Z}\cartesian\mathbb{Z}$
(resp., $\mathbb{Z}\strong\mathbb{Z}$) cannot be controlled by $f$
firefighters.

Our proofs can be easily adapted to show analoguous upper and lower
bound for the triangular grid $\mathbb{Z}\triangular\mathbb{Z}$,
in which the threshold is $2$.

\subsection{Related work}

The firefighter problem is loosely connected with Conway's angel problem~\cite{BCG82-book}.
This is a game of pursuit in $\mathbb{Z}\strong\mathbb{Z}$, in which
the angel can move to any point within $\ell_{\infty}$-distance $k$
and the devil can destroy one unoccupied point per turn, bearing similarities
to the $f\equiv1/k$ case of the firefighter problem. The two main
differences between the angel problem and the firefighter problem
are
\begin{enumerate}
\item The fire is \emph{non-deterministic}, that is, it needs not choose
its path in advance;
\item The firefighters play a \emph{predetermined} strategy, that is, they
cannot adapt their strategy to the fire's advancement.
\end{enumerate}
It is known that for $1\le k<2$, where the fractional version is
defined appropriately, the devil wins~\cite{Kutz04-phd}, and that
for $k\ge2$ the angel wins~\cite{Bowditch07,Gacs07,Kloster07,Mathe07}.
Our results, when presented as a variant of the angel problem in which
the fire is more powerful, show that the threshold is $1/3$ instead
of $2$.

\bigskip{}

The rest of the paper is organized as follows. In Section~\ref{sec:lower-bound-proof}
we prove Theorem~\ref{thm:cartesian-grid-LB}, in Section~\ref{sec:upper-bound-proof}
we prove Theorem~\ref{thm:strong-grid-UB}, and in Section~\ref{sec:corollaries-proof}
we show how these two theorems imply Corollaries~\ref{cor:strong-grid-LB}
and~\ref{cor:cartesian-grid-UB}.

Throughout the paper we denote the set of non-negative integers by
$\mathbb{N}$ and the set of integers by $\mathbb{Z}$. For a sequence
$s\left(t\right)$ we define $\lim\inf s\left(t\right)=\lim_{t_{0}\to\infty}\inf\left\{ s\left(t\right):t\ge t_{0}\right\} $.
By $\left\lceil x\right\rceil $ (resp., $\left\lfloor x\right\rfloor $)
we denote the real number $x$ rounded up (resp., down) to the closest
integer.

\section{\label{sec:lower-bound-proof}Proof of Theorem~\ref{thm:cartesian-grid-LB}}

\subsection{Time-line}

Our proof of Theorem~\ref{thm:cartesian-grid-LB} makes use of several
sequences, all of which are represented as some function measured
at integer times $t$. To circumvent ambiguity that can arise due
to timing subtleties, we define a time-line for the process as follows
(here $n$ is a positive integer).

\begin{center}
\begin{tabular}{|l|l|}
\hline 
Time $t$ & What happens?\tabularnewline
\hline
\hline 
$0$ & The grid is created, empty and void.\tabularnewline
\hline 
$1/3$ & The initial set of points $S$ is set on fire.\tabularnewline
\hline 
$n-1/3$ & The $n$th squad consisting of $f\left(n\right)$ firefighters is
placed on the grid.\tabularnewline
\hline 
$n$ & Nothing. Crickets chirp.\tabularnewline
\hline 
$n+1/3$ & The fire spreads to adjacent unprotected points.\tabularnewline
\hline
\end{tabular}
\par\end{center}

\subsection{Definitions and simple claims}

Fix a strategy using $f$ firefighters. In the following definitions
$t$ is a natural number representing time and $i,j\in\left\{ \pm1\right\} $
represent together a direction: north-east, north-west, south-west
or south-east.

Although all objects we define are a function of time, we may omit
$t$ from the notation when the context allows.

\paragraph{Fire fronts, lengths and perimeter.}

The \emph{fire front} $L_{i,j}=L_{i,j}\left(t\right)$ is the line
\[
L_{i,j}=\left\{ \left(x,y\right)\in\mathbb{Z}\times\mathbb{Z}:xi+yj=c_{i,j}\right\} ,\]
where $c_{i,j}=c_{i,j}\left(t\right)$ is the minimal natural number
for which no point on $L_{i,j}$ is burning at time $t$. %
{}

The \emph{length} $\rho_{i,j}=\rho_{i,j}\left(t\right)$ of a fire
front $L_{i,j}$ is defined as the $\ell_{\infty}$ distance between
$L_{i,j}\cap L_{i,-j}$ and $L_{i,j}\cap L_{-i,j}$.

The sum of the lengths of all four fire fronts is the \emph{fire perimeter}
at time $t$, which we denote by $\rho=\rho\left(t\right)$. Note
that $\rho_{i,j}=\rho_{-i,-j}=\frac{1}{2}\left(c_{i,-j}+c_{-i,j}\right)$
and thus $\rho=\sum_{i,j\in\left\{ \pm1\right\} }c_{i,j}$.

\paragraph{Total and front potential.}

A point is \emph{endangered} if it is unprotected and adjacent to
burning point. We define the \emph{total potential} $\phi=\phi\left(t\right)$
at time $t$ as the number of endangered points on $L\left(t\right)=\bigcup_{i,j\in\left\{ \pm1\right\} }L_{i,j}\left(t\right)$;
that is, the difference between the total number of points in $L\left(t\right)$
adjacent to burning points and the amount of firefighters protecting
such points. For consistency, we define $\phi\left(0\right)=1$ (that
is, the fire source is the single endangered point).

Note that our choice of time-line dictates that all these $\phi$
endangered points catch fire at time $t+1/3$.

Moreover, we define the \emph{potential} $\phi_{i,j}=\phi_{i,j}\left(t\right)$
of a fire front $L_{i,j}$ as the contribution of points on $L_{i,j}$
to the potential. More precisely, an endangered point on a single
$L_{i,j}$ contributes one to $\phi_{i,j}$ and an endangered point
that belongs to two adjacent fire fronts contributes $1/2$ to the
potential of each.%
\footnote{As a special case, at time $t=0$ we have $\phi_{i,j}\left(0\right)=1/4$. %
}

Note that $\phi=\sum_{i,j\in\left\{ \pm1\right\} }\phi_{i,j}$.
\begin{claim}
\label{clm:length-bounds-potential}For all $t\in\mathbb{N}$ and
$i,j\in\left\{ \pm1\right\} $ we have $\phi_{i,j}\left(t\right)\le\rho_{i,j}\left(t\right)$.\end{claim}
\begin{proof}
The length $\rho_{i,j}$ of the fire front $L_{i,j}$ must be able
to accomodate all $\phi_{i,j}$ endangered points on $L_{i,j}$, which
catch fire immediately.
\end{proof}

\paragraph{Active and frozen fronts.}

The fire front $L_{i,j}$ is \emph{active} at time $t\ge0$ if $L_{i,j}\left(t+1\right)\neq L_{i,j}\left(t\right)$
and is \emph{frozen} otherwise. Let $a_{i,j}\left(t\right)=c_{i,j}\left(t+1\right)-c_{i,j}\left(t\right)$;
that is, the indicator variable $a_{i,j}\left(t\right)$ takes the
value $1$ if $L_{i,j}\left(t\right)$ is active and the value $0$
if it is frozen.

We denote the number of active fire fronts at time $t$ by $a\left(t\right)=\sum_{i,j\in\left\{ \pm1\right\} }a_{i,j}\left(t\right)\in\left\{ 0,1,2,3,4\right\} $.
Note that by definition $a\left(t\right)=\rho\left(t+1\right)-\rho\left(t\right)$.
\begin{claim}
\label{clm:frozen-iff-zero-potential}For all $t\in\mathbb{N}$ and
$i,j\in\left\{ \pm1\right\} $ we have $a_{i,j}\left(t\right)=0$
if and only if $\phi_{i,j}\left(t\right)=0$.\end{claim}
\begin{proof}
Exactly $\phi_{i,j}\left(t\right)$ endangered points on $L_{i,j}\left(t\right)$
caught fire between time $t$ and $t+1$ (specifically, at time $t+1/3$).
The fire front is active if and only if this number is positive.
\end{proof}
Note that a reactivation of a frozen front can only occur when an
adjacent active fire front endangers its corner, giving it a potential
of $1/2$.

\subsection{Bounding the potential}

The following lemma bounds the potential from below by bounding the
change in potential between consecutive times. Denote by $f_{i,j}\left(t\right)$
the number of firefighters placed on $L_{i,j}\left(t\right)$ until
time $t$ that were not counted in any $f_{i',j'}\left(\tau\right)$
for $\tau<t$ (this distinction is needed in order to avoid double-counting
of firefighters on a frozen fire front) and let $f_{i,j}^{*}\left(t\right)=\sum_{\tau=1}^{t}f_{i,j}\left(\tau\right)$.
\begin{lem}
\label{lem:potential-LB}For all $t\in\mathbb{N}$ and $i,j\in\left\{ \pm1\right\} $
we have $\phi_{i,j}\left(t\right)\ge1/4+c_{i,j}\left(t\right)-f_{i,j}^{*}\left(t\right)$.\end{lem}
\begin{proof}
If $L_{i,j}$ is active at time $\tau$, then the $\phi_{i,j}\left(\tau\right)$
burning points on it have at least $1+\phi_{i,j}\left(\tau\right)$
neighbors in $L_{i,j}\left(\tau+1\right)$, of which at most $f_{i,j}\left(\tau+1\right)$
are protected by time $\tau+1$. If $L_{i,j}$ is frozen at time $\tau$,
then all points on $L_{i,j}\left(\tau+1\right)$ adjacent to burning
points are protected by time $\tau+1$. In any case, we have \[
\phi_{i,j}\left(\tau+1\right)\ge\phi_{i,j}\left(\tau\right)+a_{i,j}\left(\tau\right)-f_{i,j}\left(\tau+1\right).\]
Summing this for $\tau=0,1,\ldots,t-1$ yields \[
\phi_{i,j}\left(t\right)\ge\phi_{i,j}\left(0\right)-f_{i,j}^{*}\left(t\right)+\sum_{\tau=0}^{t-1}a_{i,j}\left(\tau\right)=\phi_{i,j}\left(t\right)\ge\phi_{i,j}\left(0\right)+c_{i,j}\left(t\right)-f_{i,j}^{*}\left(t\right),\]
as stated by the lemma.
\end{proof}
The next two lemmata lay the foundations for the proof of Proposition~\ref{prop:perimeter-expands}.
\begin{lem}
\label{lem:potential-exceeds-semiperimeter}If $\rho\left(t\right)\ge2f^{*}\left(t\right)-1$
then $\phi\left(t\right)>\rho\left(t\right)/2$.\end{lem}
\begin{proof}
Summed over all directions $i,j\in\left\{ \pm1\right\} $, Lemma~\ref{lem:potential-LB}
yields $\phi\left(t\right)\ge1+\rho\left(t\right)-f^{*}\left(t\right)\ge1/2+\rho\left(t\right)/2.$\end{proof}
\begin{lem}
\label{lem:sum-of-opposing-potentials-is-positive}If $\rho\left(t\right)\ge2f^{*}\left(t\right)-1$
then $\phi_{i,j}\left(t\right)+\phi_{-i,-j}\left(t\right)>0$.\end{lem}
\begin{proof}
We have $c_{i,j}\left(t\right)+c_{i,-j}\left(t\right)+c_{-i,j}\left(t\right)+c_{-i,-j}\left(t\right)=\rho\left(t\right)$
so at least one of the following cases is guaranteed to hold.
\begin{caseenv}
\item If $c_{i,j}\left(t\right)+c_{-i,-j}\left(t\right)>\rho\left(t\right)/2$,
then by applying Lemma~\ref{lem:potential-LB} twice we get \[
\phi_{i,j}\left(t\right)+\phi_{-i,-j}\left(t\right)\ge1/2+c_{i,j}\left(t\right)+c_{-i,-j}\left(t\right)-f^{*}\left(t\right)>1/2+\rho\left(t\right)/2-f^{*}\left(t\right)\ge0.\]

\item If $\rho_{i,j}\left(t\right)+\rho_{-i,-j}\left(t\right)=c_{i,-j}\left(t\right)+c_{-i,j}\left(t\right)\ge\rho\left(t\right)/2$
then by Claim~\ref{clm:length-bounds-potential} we have 
\end{caseenv}
\[
\phi_{i,-j}\left(t\right)+\phi_{-i,j}\left(t\right)\le\rho_{i,-j}\left(t\right)+\rho_{-i,j}\left(t\right)=\rho\left(t\right)-\rho_{i,j}\left(t\right)-\rho_{-i,-j}\left(t\right)\le\rho\left(t\right)/2\]
and by Lemma~\ref{lem:potential-exceeds-semiperimeter} we get \[
\phi_{i,j}\left(t\right)+\phi_{-i,-j}\left(t\right)=\phi\left(t\right)-\phi_{i,-j}\left(t\right)-\phi_{-i,j}\left(t\right)\ge\phi\left(t\right)-\rho\left(t\right)/2>0.\qedhere\]

\end{proof}
The following proposition concludes the proof by showing that the
fire expands indefinitely and thus cannot be controlled.
\begin{prop}
\label{prop:perimeter-expands}Assume that $f^{*}\left(t\right)\le\left(3t+1\right)/2$
for all $t\in\mathbb{N}$. Then $\rho\left(t\right)\ge3t$ for all
$t\in\mathbb{N}$.\end{prop}
\begin{proof}
We prove this by induction on $t$. For $t=0$ we have $\rho\left(0\right)=0$.
Assume $\rho\left(t\right)\ge3t\ge2f^{*}\left(t\right)-1$. By Lemma~\ref{lem:potential-exceeds-semiperimeter}
no two adjacent fire fronts can be frozen at time $t$, since the
sum of the potential of the two others cannot exceed the sum of their
lengths, which is the semi-perimeter. By Lemma~\ref{lem:sum-of-opposing-potentials-is-positive}
no two opposing fire fronts can be frozen at time $t$. Thus, $a\left(t\right)\ge3$
and $\rho\left(t+1\right)=\rho\left(t\right)+a\left(t\right)\ge3t+3$.
\end{proof}

\section{\label{sec:upper-bound-proof}Proof of Theorem~\ref{thm:strong-grid-UB}}

To make the proof easier, we make the following assumptions without
loss of generality.
\begin{enumerate}
\item The fire breaks out in an $\ell_{\infty}$-ball of radius $r\ge0$,
i.e., an axes-parallel $\left(2r+1\right)\times\left(2r+1\right)$
square, centered at the origin.
\item There exist some $t_{0}\in\mathbb{N}$ and $\epsilon>0$ such that
$f^{*}\left(t\right)\ge\left(3+\epsilon\right)t$ for all $t\ge t_{0}$.
This is because $\lim\inf f^{*}\left(t\right)/t>3$ implies the existence
of such $t_{0}$ and $\epsilon$ for which $\inf\left\{ f^{*}\left(t\right)/t:t\ge t_{0}\right\} \ge3+\epsilon$.
\item We may assume $t_{0}=1$ since we may enlarge the initial fire by
adding $t_{0}$ to $r$.
\end{enumerate}
The only property of $f$ we will use, which is a strengthened form
of $f^{*}\left(t\right)>3t$, is the following. Set $m=\left\lceil 1/\epsilon\right\rceil $.
Then for all $k\in\mathbb{N}$ we have \[
f^{*}\left(mk+1\right)\ge3\left(mk+1\right)+\epsilon\left(mk+1\right)>3mk+3m+k.\]

Now we describe a strategy $\mathcal{S}=\mathcal{S}\left(m,r\right)$
that allows $f$ firefighters to control a fire that breaks out in
an $\ell_{\infty}$-ball of radius $r\ge1$ centered at the origin.

Our strategy has four phases. In a terminology similar to the one
used in Section~\ref{sec:lower-bound-proof}, we are guaranteed to
have at least $k-1$ frozen fronts during the $k$th phase, hence
when the fourth phase ends, all four fronts are frozen and the fire
is controlled. 

The following invariants are maintained:
\begin{itemize}
\item The shape of the fire at all times is an $\ell_{\infty}$-ellipse
(that is, an axes-parallel rectangle).
\item The firefighters are placed on the perimeter of an $\ell_{\infty}$-ellipse.
\item Each firefighter is placed next to an already positioned firefighter
(except for the first one, of course).
\end{itemize}
\noindent \begin{center}
\begin{table}[h]
\noindent \begin{centering}
\begin{tabular}{|c|c|c|c|c|}
\hline 
Time & Fire width & Fire height & Available firefighters & Frozen\tabularnewline
\hline
\hline 
$1$ & $2r+1$ & $2r+1$ & $\ge4$ & -\tabularnewline
\hline 
$2r$ & $6r-1$ & $6r-1$ & $\ge6r+1$ & North\tabularnewline
\hline 
$6rm+1$ & $12rm+2r+1$ & $6rm+4r$ & $\ge18rm+6r+4$ & North, East\tabularnewline
\hline 
$6rm^{2}+10rm$ & $6rm^{2}+16rm+2r$ & $6rm^{2}+10rm+3r-1$ & $\ge18rm^{2}+36rm+10r$ & N, E, W\tabularnewline
\hline 
$12rm^{2}+30rm$ & $6rm^{2}+16rm+2r$ & $12rm^{2}+30rm+3r-1$ & $\ge36rm^{2}+102rm+30r$ & All\tabularnewline
\hline
\end{tabular}
\par\end{centering}

\caption{Key times for the strategy $\mathcal{S}\left(m,r\right)$}

\end{table}

\par\end{center}

\paragraph{First phase: northern front.}

This phase begins at time $t=1$ and ends at time $t=2r$. All firefighters
are placed on the horizontal line $y=3r$ between $x_{\min}=1-3r$
and $x_{\max}=3r-1$. Note that $f^{*}\left(2r\right)\ge6r+1$ and
thus by the end of the phase, when the fire has grown to an $\ell_{\infty}$-ball
of radius $3r-1$ and has reached the northern front, the front is
just long enough so that the fire is not able to spread north anymore.

\paragraph{Second phase: eastern front.}

This phase begins at time $t=2r+1$ and ends at time $t=6rm+1$. While
maintaining the west end of the northern front just out of the fire's
reach, the firefighters continue the northern front eastwards until
$x_{\max}=6rm+r+1$ and build an eastern front on this vertical line,
starting at the corner $y_{\max}=3r$ and going south until $y_{\min}=-r-6rm$.
Note that $f^{*}\left(6rm+1\right)\ge18rm+6r+4$ and thus by the end
of the phase, when the fire has grown to an $\ell_{\infty}$-ellipse
of height $6rm+4r$ and width $12rm+2r+1$ and has reached the eastern
front, the front is just long enough so that the fire is not able
to spread east anymore.

\paragraph{Third phase: western front.}

This phase begins at time $t=6rm+2$ and ends at time $t=6rm^{2}+10rm$.
While maintaining the south end of the eastern front just out of the
fire's reach, the firefighters continue the northern front westwards
until $x_{\min}=1-r-10rm-6rm^{2}$ and build a western front on this
vertical line, starting at the corner $y_{\max}=3r$ and going south
until $y_{\min}=1-10rm-6rm^{2}$. Note that $f^{*}\left(6rm^{2}+10rm\right)\ge18rm^{2}+36rm+10r$
and thus by the end of the phase, when the fire has grown to an $\ell_{\infty}$-ellipse
of height $6rm^{2}+10rm+3r-1$ and width $6rm^{2}+16rm+2r$ and has
reached the western front, the front is long enough (by $2r$ or so)
so that the fire is not able to spread east anymore.

\paragraph{Fourth phase: southen front.}

This phase begins at time $t=6rm^{2}+10rm+1$ and ends at time $t=12rm^{2}+30rm$.
While maintaining the south end of the eastern front just out of the
fire's reach, the firefighters continue the western front southwards
until $y_{\min}=1-30rm-12rm^{2}$ and build a southern front on this
horizontal line, starting at the corner $x_{\min}=1-r-10rm-6rm^{2}$
and going east until the eastern front is met at the corner $x_{\max}=6rm+r+1$.
This actually happens about $20r$ rounds before the end of the phase,
since $f^{*}\left(12rm^{2}+30rm\right)\ge36rm^{2}+102rm+30r$ and
thus by the end of the phase, when the fire has grown to an $\ell_{\infty}$-ellipse
of height $12rm^{2}+30rm+3r-1$ and width $6rm^{2}+16rm+2r$, it is
fully surrounded.

\section{\label{sec:corollaries-proof}Proof of Corollaries~\ref{cor:strong-grid-LB}
and~\ref{cor:cartesian-grid-UB}}

Using the following proposition, Theorem~\ref{thm:cartesian-grid-LB}
implies Corollary~\ref{cor:strong-grid-LB} and Theorem~\ref{thm:strong-grid-UB}
implies Corollary~\ref{cor:cartesian-grid-UB}.
\begin{prop}
If $f$ firefighters can control a fire that breaks out in a ball
of radius $r\ge0$ in the strong grid $\mathbb{Z}\strong\mathbb{Z}$,
then $g$ firefighters can control a fire that breaks out in a ball
of radius $2r$ in the Cartesian grid $\mathbb{Z}\cartesian\mathbb{Z}$,
where the function $g$ is defined by $g\left(t\right)=\begin{cases}
\left\lfloor f\left(k\right)/2\right\rfloor , & \textrm{if }t=2k-1;\\
\left\lceil f\left(k\right)/2\right\rceil , & \textrm{if }t=2k.\end{cases}$\end{prop}
\begin{proof}
Without loss of generality, the fire center is the origin in both
grids. Let $\mathcal{S}$ be the strategy used by the firefighters
in $\mathbb{Z}\strong\mathbb{Z}$ to control the fire, and assume
that at time $t$, firefighters are placed in a set $P_{t}$ of $\left|P_{t}\right|=f\left(t\right)$
points.

\global\long\def\rot{\Lsh}
We exploit the connection between the metrics $\ell_{1}$ and $\ell_{\infty}$
on the plane $\mathbb{R}^{2}$ to convert $\mathcal{S}$ to a strategy
$\mathcal{S}'$ for $\mathbb{Z}\cartesian\mathbb{Z}$. Specifically,
we use the injective mapping $\rot:\mathbb{Z}\times\mathbb{Z}\to\mathbb{Z}\times\mathbb{Z}$
defined by $\rot\left(x,y\right)=\left(x+y,x-y\right)$. Partition
the set $P_{t}$ arbitrarily to two sets $P_{t}'$ and $P_{t}^{''}$
of respective sizes $g\left(2t-1\right)$ and $g\left(2t\right)$.
It is possible as $\left|P_{t}\right|=f\left(t\right)=g\left(2t-1\right)+g\left(2t\right)$.
The strategy $\mathcal{S}'$ places firefighters in $\rot\left(P_{t}'\right)$
at time $2t-1$ and in $\rot\left(P_{t}''\right)$ at time $2t-1$. 

Note that $\mathcal{S}'$ only places firefighters at even points;
that is, points $\left(x,y\right)$ such that $x+y$ is even. Recall
that the graph $\mathbb{Z}\cartesian\mathbb{Z}$ is bipartite, and
the initial fire boundary consists of even points only. Therefore,
at odd times the fire can only spreads to odd points (which are never
protected) and at even times the fire can spread only to unprotected
even points. It makes sense thus to consider the state of the process
only at even times $t=2k$. But behold --- the square of the graph
$\mathbb{Z}\cartesian\mathbb{Z}$ restricted to even points is isomorphic
to $\mathbb{Z}\strong\mathbb{Z}$ using the isomorphism $\rot$, and
the initial fire, the $\ell_{1}$-ball of radius $r$, is mapped by
$\rot$ to an $\ell_{\infty}$-ball of radius $2r$.

Since the strategy $\mathcal{S}$ is able to control the fire in $\mathbb{Z}\strong\mathbb{Z}$
in some finite time $T$, the strategy $\mathcal{S}'$ will control
the fire in the even part of $\mathbb{Z}\cartesian\mathbb{Z}$. This
establishes the result as $\mathbb{Z}\cartesian\mathbb{Z}$ is bipartite.
\end{proof}

\subsection*{Acknowledgements}

The authors thank Noga Alon for useful discussions.

\end{document}